\documentclass[runningheads]{llncs}
\usepackage{graphicx}
\usepackage{underscore}           
\usepackage{amsmath}
\usepackage{amssymb}
\usepackage{enumerate}
\usepackage{verbatim}
\usepackage{mathtools}

\renewcommand{\phi}{\varphi}
\newcommand*{\Nset}{\mathbb{N}}

\newcommand*{\size}{\mathrm{sz}}

\newcommand*{\FO}{\mathrm{FO}}
\newcommand*{\MSO}{\mathrm{MSO}}

\newcommand*{\tower}{\mathrm{twr}}
\newcommand*{\Ntower}{\mathrm{tower}}

\newcommand*{\oset}{\mathtt{oset}}

\newcommand*{\set}{\mathtt{set}}
\newcommand*{\setc}{\mathtt{core}}

\newcommand*{\lev}{\mathtt{levels}}
\newcommand*{\add}{\mathtt{add}}

\newcommand*{\EF}{\text{Ehrenfeucht-Fra\"{\i}ss\'{e}}}

\newcommand{\qrfo}{\mathrm{qr}}

\newcommand{\mlm}{\mu}
\newcommand{\Mlm}{\nu}
\newcommand{\defN}{\mathrm{DN}}
\newcommand{\Def}{\mathrm{Def}}

\newcommand{\LS}{\mathrm{LS}}
\newcommand{\Hanf}{\mathrm{H}}
\newcommand{\tp}{\mathrm{tp}}

\newtheorem{innercustomthm}{Theorem}
\newenvironment{customthm}[1]
  {\renewcommand\theinnercustomthm{#1}\innercustomthm}
  {\endinnercustomthm}

\begin{document}
\title{Defining long words succinctly in FO and MSO\thanks{Miikka Vilander acknowledges the financial support of the Academy of Finland project \emph{Explaining AI via Logic} (XAILOG), project number 345612.}}
%
%
\author{Lauri Hella\orcidID{0000-0002-9117-8124} \and
Miikka Vilander\orcidID{0000-0002-7301-939X}}
\authorrunning{L. Hella and M. Vilander}
%
\institute{Tampere University, 33100 Tampere, Finland}
\maketitle              
\begin{abstract}
We consider the length of the longest word definable in FO and MSO via a formula of size $n$. For both logics we obtain as an upper bound for this number an exponential tower of height linear in $n$. We prove this by counting types with respect to a fixed quantifier rank. As lower bounds we obtain for both FO and MSO an exponential tower of height in the order of a rational power of $n$. We show these lower bounds by giving concrete formulas defining word representations of  levels of the cumulative hierarchy of sets. In addition, we consider the Löwenheim-Skolem and Hanf numbers of these logics on words and obtain similar bounds for these as well.

\end{abstract}
\section{Introduction}

We consider the succinctness of defining words. More precisely, if we allow formulas of size up to $n$ in some logic, we want to know the length of the longest word definable by such formulas. 

This question is not very interesting for all formalisms. An example where this is the case is given by regular expressions. There is no smaller regular expression that defines a word than the word itself. This result is spelled out at least in the survey \cite{EllulKSW2005}. However, the situation is completely different for monadic second-order logic MSO, even though MSO has the same expressive power as regular expressions over words. In this paper we consider the definability of words in MSO as well as first-order logic FO. We show that these logics can define words of non-elementary length via formulas of polynomial size.

In order to argue about definability via formulas of bounded size, we define the size $n$ fragments $\FO[n]$ and $\MSO[n]$ that include only formulas of size up to $n$. We also define similar quantifier rank $k$ fragments $\FO_k$ and $\MSO_k$ and use them to prove our upper bounds. Both of these types of fragments are essentially finite in the sense that they contain only a finite number of non-equivalent formulas. We call the length of the longest word definable in a fragment the definability number of that fragment. Using this concept, our initial question is reframed as studying the definability numbers of $\FO[n]$ and $\MSO[n]$.

The definability number of a fragment is closely related to the Löwenheim-Skolem and Hanf numbers of the fragment. The Löwenheim-Skolem number of a fragment is the smallest number $m$ such that each satisfiable formula in the fragment has a model of size at most $m$. The Hanf number is the smallest number $l$ such that any formula with a model of size greater than $l$ has arbitrarily large models. These were originally defined for extensions of first-order logic in the context of model theory of infinite structures, but they are also meaningful in the context of finite structures. For a survey on Löwenheim-Skolem and Hanf numbers both on infinite and finite structures see \cite{EbbSurvey}. For previous research on finite Löwenheim-Skolem type results see \cite{Grohe96} and \cite{Grohe02}.

Aside from what we have already mentioned, related work includes the article \cite{PikhurkoV05} of Pikhurko and Verbitsky, where they consider the complexity of single finite structures. They study the minimal quantifier rank in FO of both defining a single finite structure and separating it from other structures of the same size. In \cite{PikhurkoV09} the same authors survey the logical complexity of single graphs in FO. By logical complexity they mean minimal quantifier rank, number of variables and length of a defining formula as functions of the size of the graph. They give an extensive account of these measures and relate them to each other, the Ehrenfeucht-Fraïssé game and the Weisfeiler-Lehman algorithm. An important difference between our approach and theirs is that we take formula size as the parameter and look for the longest definable word, whereas they do the opposite. 


Our contributions are upper and lower bounds for the definability, Löwenheim-Skolem and Hanf numbers of the size $n$ fragments of FO and MSO on words. The upper bounds in section \ref{upper} are obtained by counting types with respect to the quantifier rank $n/2$ fragment. The upper bounds for both FO and MSO are exponential towers of height $n/2+\log^*(t)+1$ where $t$ is a polynomial term. The lower bounds in sections \ref{fo} and \ref{mso} are given by concrete polynomial size formulas that define words of non-elementary length. The lower bounds are exponential towers of height $\sqrt[5]{n/c}$ for FO and $\sqrt{n/c}$ for MSO, respectively. 

Note that our results only apply in the context of words. If finite structures over arbitrary finite vocabularies are allowed, then there are no computable upper bounds for the Löwenheim-Skolem or Hanf numbers of the size $n$ fragments of FO. For the Löwenheim-Skolem number, this follows from Trakhtenbrot's theorem\footnote{Trakhtenbrot's theorem states that the finite satisfiability problem of FO is undecidable. Hence there cannot exist any computable upper bound for the size of models that need to be checked to see whether a given formula is satisfiable.} (see, e.g., \cite{Libkin}), and for the Hanf number, this follows from a result of Grohe in \cite{Grohe96}. Clearly the same applies for the size $n$ fragments of MSO as well.

\section{Preliminaries}

The logics we consider in this paper are first-order logic FO and monadic second-order logic MSO and their (typically finite) fragments. The syntax and semantics of these are standard and well-known. Due to space restrictions we will not present them here, instead directing the reader to \cite{EbbFlum} and \cite{Libkin}.

In terms of structures we limit our consideration to words of the two letter alphabet $\Sigma = \{l, r\}$. We interpret these symbols as left and right parentheses but use letters for easier readability. When we say that a word satisfies a logical sentence, we mean the natural corresponding word model does. A word model is a finite structure with linear order and unary predicates $P_l$ and $P_r$ for the two symbols. 

Since we only consider words over the two letter alphabet $\Sigma$, we will tacitly assume that all formulas of $\MSO$ are in the vocabulary $\{<,P_l,P_r\}$ of the corresponding word models (and similarly for $\FO$-formulas).

\begin{definition}
The size $\size(\phi)$ of a formula $\phi \in \MSO$ is defined recursively as follows:
\begin{itemize}
    \item $\size(\phi) = 1$ for atomic $\phi$,
    \item $\size(\psi \land \theta) = \size(\psi \lor \theta) = \size(\psi) + \size(\theta) + 1$,
    \item $\size(\exists x \psi) = \size(\forall x \psi) = \size(\exists U \psi) = \size(\forall U \psi) = \size(\neg \psi) = \size(\psi) + 1$.
\end{itemize}
For $n \in \Nset$ the size $n$ fragment of $\MSO$, denoted $\MSO[n]$, consists of the formulas of $\MSO$ with size at most $n$. Size as well as size $n$ fragments are defined in the same way for $\FO$.
\end{definition}

\begin{definition}
The quantifier rank $\qrfo(\phi)$ of a formula $\phi \in \MSO$ is defined recursively as follows:
\begin{itemize}
    \item $\qrfo(\phi) = 0$ for atomic $\phi$,
    \item $\qrfo(\neg \psi) = \qrfo(\psi)$,
    \item $\qrfo(\psi \land \theta) = \qrfo(\psi \lor \theta) = \max\{\qrfo(\psi),\qrfo(\theta)\}$,
    \item $\qrfo(\exists x \psi) = \qrfo(\forall x \psi) = \qrfo(\exists U \psi) = \qrfo(\forall U \psi) = \qrfo(\psi) + 1$.
\end{itemize}
For $k \in \Nset$, the quantifier rank $k$ fragment of $\MSO$, denoted $\MSO_k$, consists of the formulas $\phi \in \MSO$ with $\qrfo(\phi) \leq k$. The quantifier rank $k$ fragment of $\FO$ is defined in the same way and denoted $\FO_k$.
\end{definition}

Note that both size $n$ fragments and quantifier rank $k$ fragments are essentially finite in the sense that they contain only finitely many non-equivalent formulas.

\begin{definition}
For each (finite) fragment $L$ of $\MSO$ or $\FO$, we define the relation $\equiv_L$ on $\Sigma$-words as
\[
w \equiv_L v, \text{ if $w$ and $v$ agree on all $L$-sentences}.
\]
Clearly $\equiv_L$ is an equivalence relation. We denote the number of equivalence classes of $\equiv_L$ on $\Sigma$-words by $N_L$. 
\end{definition}

Note that each equivalence class of $\equiv_L$ is uniquely determined by a subset $\tp_L(w)=\{\varphi\in L\mid w\models\varphi\}$ of $L$ sentences, which we call the $L$-type of $w$. Thus, $N_L$ is the number of $L$-types. In the case $L=\MSO_k$ or $L=\FO_k$, we talk about quantifier rank $k$ types.

\begin{definition}
We say that a sentence $\phi \in \MSO$ defines a word $w \in \Sigma^+$ if $w \vDash \phi$ and $v \nvDash \phi$ for all $v \in \Sigma^+ \setminus \{w\}$.

For a fragment $L$ of $\MSO$ or $\FO$, we denote by $\Def(L)$ the set of words definable in $L$, i.e.
\[
\Def(L) := \{w \in \Sigma^+ \mid \text{there is } \phi \in L \text{ s.t. $\phi$ defines $w$\}.} 
\]
\end{definition}

\begin{definition}
The exponential tower function $\Ntower: \Nset \to \Nset$ is defined recursively by setting $\Ntower(0) := 1$ and $\Ntower(n+1) := 2^{\Ntower(n)}$. We extend this definition to a function $\tower: [0,\infty[ \to \Nset$ by setting $\tower(x) = \Ntower( \lceil x \rceil)$. The iterated logarithm function $\log^*: [1,\infty[ \to \Nset$ is defined by setting $\log^*(x)$ as the smallest $m \in \Nset$ that has $\Ntower(m) \geq x$.

\end{definition}

\subsection{Definability, Löwenheim-Skolem and Hanf numbers}


Löwenheim-Skolem and Hanf numbers were originally introduced for studying the behaviour of extensions of first-order logic on infinite structures. See the article \cite{EbbSurvey} of Ebbinghaus for a nice survey on the infinite case. As observed in \cite{Grohe96}, with suitable modifications, it is possible to give meaningful definitions for these numbers also on finite structures. We will now give such definitions for finite fragments $L$ of $\FO$ and $\MSO$, and in addition, we introduce the closely related definability number of $L$.

Let $\phi$ be a sentence in $\MSO$ over $\Sigma$-words. If it has a model, we denote by $\mlm(\phi)$ the minimal length of a model of $\phi$: $\mlm(\phi)=\min\{|w|\mid w \in \Sigma^+, w\models\phi\}$. If $\phi$ has no models, we stipulate $\mlm(\phi) = 0$. Furthermore, we denote by $\Mlm(\phi)$ the maximum length of a model of $\phi$, assuming the maximum is well-defined. If the maximum is not defined, i.e., if $\phi$ has no models or has arbitrarily long models, we stipulate $\Mlm(\phi)=0$.

\begin{definition}\label{LS-Hanf}
Let $L$ be a finite fragment of $\MSO$ or $\FO$ with $\Def(L) \neq \emptyset$.

(a) The definability number of $L$ is \\ \phantom{a} \hfill $\defN(L)=
    \max\{|w|\mid w \in \Sigma^+, w\in\Def(L)\}$. \hfill \phantom{a}

(b) The Löwenheim-Skolem number of $L$ is $\LS(L)=\max\{\mlm(\phi)\mid \phi\in L\}$.

(c) The Hanf number of $L$ is $\Hanf(L)=\max\{\Mlm(\phi)\mid \phi\in L\}$.
\end{definition}

Thus, $\defN(L)$ is the length of the longest $L$-definable word. Note further that $\LS(L)$ is the smallest number $m$ such that every $\phi\in L$ that has a model, has a model of length at most $m$. Similarly $\Hanf(L)$ is the smallest number $\ell$ such that if $\phi\in L$ has a model of length greater than $\ell$, then it has arbitrarily long models.

Since every sentence $\phi$ of $\MSO$ defines a regular language over $\Sigma$, and there is an effective translation from $\MSO$ to equivalent finite automata, it is clear that we can compute the numbers $\mlm(\phi)$ and $\Mlm(\phi)$ from $\phi$. Consequently, for any finite fragment $L$ of $\MSO$, $\LS(L)$ and $\Hanf(L)$ can be computed from $L$.

As we mentioned in the Introduction, $\LS(\FO[n])$ and $\Hanf(\FO[n])$ are not computable from $n$ if we consider arbitrary finite models instead of words. Clearly the same holds also for the fragments $\FO_k$, $\MSO[n]$ and $\MSO_k$. 

It follows immediately from Definition \ref{LS-Hanf} that the definability number of any finite fragment of $\MSO$ is bounded above by its Löwenheim-Skolem number and its Hanf number:

\begin{proposition}\label{DNvsLS}
If $L$ is finite fragment of $\MSO$, then 
$\defN(L)\le \LS(L),\Hanf(L)$.
\end{proposition}
\begin{proof}
It suffices to observe that if $w \in\Def(L)$, then $\mlm(\phi)=\Mlm(\phi)=|w|$, where $\phi \in L$ is the formula that defines $w$.
\end{proof}

\section{Upper bounds for the length of definable words}\label{upper}

\subsection{Definability and types}

It is well-known that equivalence of words up to a quantifier rank is preserved in catenation:

\begin{theorem}\label{preservation}
Let $L\in\{\FO_k,\MSO_k\}$ for some $k\in\Nset$. Assume that $v,v',w,w'\in \Sigma^+$ are words such that $v\equiv_L v'$ and $w\equiv_L w'$. Then $v w\equiv_L v' w'$.
\end{theorem}
\begin{proof}
The claim is proved by a straightforward $\EF$\ game argument (see Proposition 2.1.4 in \cite{EbbFlum}).
\end{proof}


Using Theorem \ref{preservation}, we get the following upper bounds for the numbers $\mlm(\phi)$ and $\Mlm(\phi)$ in terms of the quantifier rank of $\phi$:

\begin{proposition}\label{mlm-bound}
Let $L\in\{\FO_k,\MSO_k\}$ for some $k\in\Nset$. If $\phi$ is a sentence of $L$, then $\mlm(\phi),\Mlm(\phi)\le N_L$.
\end{proposition}

\begin{proof}
If $|w|\le N_L$ for all words $w \in \Sigma^+$ such that $w\models\phi$, the claim is trivial.
Assume then that $w\models\phi$ and $|w|>N_L$. Then there are two initial segments $u$ and $u'$ of $w$ such that $|u|<|u'|$ and $u\equiv_L u'$. Let $v$ and $v'$ be the corresponding end segments, i.e., $w=uv=u'v'$. Then by Theorem \ref{preservation}, $uv'\equiv_L u'v'=w$, and similarly $u'v\equiv_L uv=w$, whence $uv'\models\phi$ and $u'v\models\phi$. 

Since $|uv'|<|w|$, we see that $w$ is not the shortest word satisfying $\phi$. The argument applies to any word $w$ with $|w|>N_L$, whence we conclude that $\mlm(\phi)\le N_L$. On the other hand $|u'v|>|w|$, whence $w$ is neither the longest word satisfying $\phi$. Applying this argument repeatedly, we see that $\phi$ is satisfied in arbitrarily long words, whence $\Mlm(\phi)=0\le N_L$. 
\end{proof}

From Propositions \ref{DNvsLS} and \ref{mlm-bound} we immediately obtain the following upper bound for the definability numbers of quantifier rank fragments of $\MSO$:

\begin{corollary}\label{DNLleNL}
Let $k\in\Nset$ and $L\in\{\FO_k,\MSO_k\}$.
Then  $\LS(L),\Hanf(L)\le N_L$, and consequently $\defN(L)\le N_L$.
\end{corollary}

This $N_L$ upper bound for the definability, Löwenheim-Skolem and Hanf numbers shows that the quantifier rank fragments $L$ of $\FO$ and $\MSO$ behave quite tamely on words: Clearly every type $\tp_L(w)$ is definable by a sentence of $L$, whence the number of non-equivalent sentences in $L$ is $2^{N_L}$. Thus, any collection of representatives of non-equivalent sentences of $L$ necessarily contains sentences of size close to $N_L$. But in spite of this, it is not possible to define words that are longer than $N_L$ by sentences of $L$.

This shows that quantifier rank is not a good starting point if we want to prove interesting succinctness results for definability. Hence we turn our attention to the size $n$ fragments $\FO[n]$ and $\MSO[n]$.
Note first that for any $n \in \Nset$, $\FO[n]$ is trivially contained in $\FO_n$, and similarly, $\MSO[n]$ is contained in $\MSO_n$. A simple argument shows that this can be improved by a factor of 2:

\begin{lemma}\label{2n-lemma}
For any $n\in\Nset$, $\FO[2n]\le\FO_n$ and $\MSO[2n]\le\MSO_n$.
\end{lemma}
\begin{proof}
(Idea) Any sentence $\phi$ with quantifier rank $n$ is equivalent to one with smaller quantifier rank unless it contains atomic formulas of the form $x<y$ mentioning each quantified variable, and more than one of them at least twice. Counting the quantifiers, the atomic formulas, and the connectives needed, we see that $\size(\phi)\ge 2n$.
\end{proof}

Note that we have not tried to be optimal in the formulation of Lemma \ref{2n-lemma}. We believe that with a more careful analysis, $2n$ could be replaced with $3n$, and possibly with an even larger number.

\begin{corollary}\label{sizeNL}
For any $n\in\Nset$,   $\defN(\FO[2n]),\LS(\FO[2n]),\Hanf(\FO[2n])\le N_{\FO_n}$ and $\defN(\MSO[2n]),\LS(\MSO[2n]),\Hanf(\MSO[2n])\le N_{\MSO_n}$. 
\end{corollary}

\subsection{Number of types}\label{number-of-types}

As we have seen in the previous section, the numbers of $\FO_k$-types and $\MSO_k$-types give upper bounds for the corresponding definbability, Löwenheim-Skolem and Hanf-numbers. It is well known that on finite relational structures, for $\FO_k$ this number is bound above by an exponential tower of height $k+1$ with a polynomial, that depends on the vocabulary, on top (see, e.g., \cite{PikhurkoV09} for the case of graphs). It is straightforward to generalize this type of upper bound to $\MSO_k$. On the class of $\Sigma$-words, we can prove the following explicit upper bounds. For the proof of this result, see the Appendix.

\begin{theorem}
For any $k\in\Nset$, 
$N_{\FO_k}\le \tower(k+\log^*(k^2+k)+1)$\\ and
$N_{\MSO_k}\le \tower(k+\log^*((k+1)^2)+1)$.
\end{theorem}

By Corollary \ref{DNLleNL}, we obtain the same upper bounds for the definability, Löwenheim-Skolem and Hanf numbers of the quantifier rank fragments.

\begin{corollary}
For any $k\in\Nset$,\\ $\defN(\FO_k),\LS(\FO_k),\Hanf(\FO_k)\le\tower(k+\log^*(k^2+k)+1)$ and \\ $\defN(\MSO_k),\LS(\MSO_k),\Hanf(\MSO_k)\le\tower(k+\log^*((k+1)^2)+1)$.
\end{corollary}

As we discussed after Corollary \ref{DNLleNL}, from the point of view of succinctness it is more interesting to consider the definability numbers of the size fragments of $\FO$ and $\MSO$ than those of the quantifier rank fragments. Using Corollary \ref{sizeNL}, we obtain the following upper bounds for $\FO[n]$ and $\MSO[n]$.

\begin{corollary}
For any $n\in\Nset$,\\
$\defN(\FO[n]), \LS(\FO[n]), \Hanf(\FO[n]) \leq \tower(n/2+\log^*((n/2)^2+n/2)+1)$
and \\ $\defN(\MSO[n]), \LS(\MSO[n]), \Hanf(\MSO[n]) \leq \tower(n/2+\log^*((n/2+1)^2)+1)$.

\end{corollary}

In the next two sections we will prove lower bounds for the definability numbers of $\FO[n]$ and $\MSO[n]$ by providing explicit polynomial size sentences that define words that are of exponential tower length. 

\section{Lower bounds for FO}\label{fo}

In order to obtain a lower bound for $\defN(\FO[n])$ we need a relatively small $\FO$-formula that defines a long word. The long word we define has to do with the cumulative hierarchy of finite sets. 

Consider representing finite sets using only braces $\{$ and $\}$. This gives each set multiple encodings as words. For the finite levels $V_i$ of the cumulative hierarchy of sets, such words clearly have length at least $\tower(i)$. We will define one such word via an $\FO$-formula of polynomial size with respect to $i$.

For readability, we define $L(x) := P_l(x)$ and $R(x) := P_r(x)$ that say $x$ is a left or right brace, respectively. We also define $S(x,y) := x<y \land \neg \exists z (x<z<y)$ that says $y$ is the successor of $x$.

As each set in the encoding can be identified by its outmost braces, the formula mostly operates on pairs of variables. For readability we adopt the convention $\overline{x} := (x_1, x_2)$, and similarly for different letters, to denote these pairs. To ensure that our formula defines a single encoding of $V_i$, we also define a linear order on encoded sets and require that the elements are in that order.

We define our formula recursively in terms of many subformulas. We briefly list the meanings and approximate sizes of each subformula involved:
\begin{itemize}
    \item $\setc (\overline{x}, \theta(s,t))$: the common core formula used in $\set_i$ and $\oset_i$. This is defined only to save space. The variables $s$ and $t$ are used only to make the formula smaller.
    \begin{align*}
        \setc (\overline{x}, \theta(s,t)) &:= x_1<x_2 \land L(x_1) \land R(x_2) \\
        &\land \forall y (x_1 < y < x_2 \rightarrow \exists z (x_1<z<x_2 \land y \neq z \\
        &\land \exists s\exists t((y < z \rightarrow (s = y \land t = z)) \\
        &\land (z < y \rightarrow (s = z \land t = y)) \land \theta(s,t))))
    \end{align*}
    \item $\set_i(\overline{x})$: $\overline{x}$ correctly encodes a set in $V_i$, possibly with repetition. Size linear in $i$. 
    \begin{align*}
        \set_0(\overline{x}) &:= L(x_1) \land R(x_2) \land S(x_1,x_2) \\
        \set_{i+1}(\overline{x}) &:= \setc(\overline{x}, \set_i(s,t))
    \end{align*}
    \item $\overline{x} \in_i \overline{y}$: $\overline{x}$ is an element of $\overline{y}$. Size linear in $i$. Assumes that $\overline{x}$ encodes a set in $V_i$ and $\overline{y}$ encodes a set in $V_{i+1}$. 
    \begin{align*}
        \overline{x} \in_i \overline{y} &:= y_1 < x_1 < x_2 < y_2 \land \neg\exists \overline{z} (\set_i(\overline{z}) \land y_1<z_1<x_1 \land x_2<z_2<y_2)
    \end{align*}
    \item $\overline{x} \sim_i \overline{y}$: $\overline{x}$ and $\overline{y}$ encode the same set. Size $\mathcal{O}(i^2)$. Assumes $\overline{x}$ and $\overline{y}$ encode sets in $V_i$.
    \begin{align*}
        \overline{x} \sim_0 \overline{y} &:= \top \\
        \overline{x} \sim_{i+1} \overline{y} &:= \forall \overline{a}(\set_i(\overline{a}) \rightarrow \exists \overline{b} (\set_i(\overline{b}) \\
        &\land (\overline{a} \in_i \overline{x} \rightarrow \overline{b} \in_i \overline{y}) \land (\overline{a} \in_i \overline{y} \rightarrow \overline{b} \in_i \overline{x}) \land \overline{a} \sim_i \overline{b}))
    \end{align*}
    \item $\overline{x} \prec_i \overline{y}$: the $\prec_{i-1}$-greatest element of the symmetric difference of $\overline{x}$ and $\overline{y}$ is in $\overline{y}$. Size $\mathcal{O}(i^3)$. Defines a linear order for encoded sets in $V_i$.
    \begin{align*}
        \overline{x} \prec_0 \overline{y} &:= \bot \\
        \overline{x} \prec_{i+1} \overline{y} &:= \exists \overline{z} (\set_i(\overline{z}) \land \overline{z} \in_i \overline{y} \land \forall \overline{a} ((\set_i(\overline{a}) \land \overline{a} \in_i \overline{x}) \\
        \rightarrow &(\overline{a} \nsim_i \overline{z} \land (\forall \overline{b} ((\set_i(\overline{b}) \land \overline{b} \in_i \overline{y}) \rightarrow \overline{a} \nsim_i \overline{b}) \rightarrow \overline{a} \prec_i \overline{z}))))
    \end{align*}
    \item $\oset_i(\overline{x})$: $\overline{x}$ correctly encodes a set in $V_i$ with no repetition and with the elements in the linear order given by the formula $\overline{x} \prec_i \overline{y}$. Size $\mathcal{O}(i^4)$. Ensures that only a singular word satisfies our formula.
    \begin{align*}
        \oset_0(\overline{x}) &:= L(x_1) \land R(x_2) \land S(x_1,x_2) \\
        \oset_{i+1}(\overline{x}) &:= \setc(\overline{x}, \oset_i(s,t))
        \land \forall \overline{a} \forall \overline{b} ((\set_i(\overline{a}) \land \set_i(\overline{b}) \\ &\land \overline{a} \in_i \overline{x} \land \overline{b} \in_i \overline{x} \land a_1<b_1) \rightarrow \overline{a} \prec_{i} \overline{b})
    \end{align*}
    \item $\add_i(\overline{x},\overline{y},\overline{z})$: States that $\overline{x} = \overline{y} \cup \{\overline{z}\}$. Size $\mathcal{O}(i^2)$. Assumes $\overline{x}$ and $\overline{y}$ encode sets in $V_i$ and $\overline{z}$ encodes a set in $V_{i-1}$. 
    \begin{align*}
        \add_{i+1}(\overline{x}, \overline{y}, \overline{z}) &:= \forall \overline{a}((\set_{i}(\overline{a}) \land \overline{a} \in_i \overline{y}) \rightarrow \exists \overline{b}(\set_{i}(\overline{b}) \land \overline{b} \in_i \overline{x} \land \overline{a} \sim_{i} \overline{b})) \\
        &\land \exists \overline{c}(\set_{i}(\overline{c}) \land \overline{c} \in_i \overline{x} \land \overline{c} \sim_{i} \overline{z} \\
        &\land \forall \overline{d}((\set_{i}(\overline{d}) \land \overline{d} \in_i \overline{x} \land d_1 \neq c_1) \\
        &\rightarrow \exists \overline{e}(\set_{i}(\overline{e}) \land \overline{e} \in_i \overline{y} \land \overline{e} \sim_{i} \overline{d})))
    \end{align*}
    \item $V_i(\overline{x})$: $\overline{x}$ encodes the set $V_i$. Size $\mathcal{O}(i^5)$. 
        \begin{align*}
            V_0(\overline{x}) &:= \set_0(\overline{x}) \\
            V_{i+1}(\overline{x}) &:= \oset_{i+1}(\overline{x}) \land \exists \overline{a} (V_0(\overline{a}) \land S(x_1,a_1)) \land \exists \overline{b} (V_i(\overline{b}) \land S(b_2,x_2) \\
            &\land \forall \overline{c} \forall \overline{d}((\set_i(\overline{c}) \land \overline{c} \in_i \overline{x} \land \set_{i-1}(\overline{d}) \land \overline{d} \in_{i-1} \overline{b}) \\
            &\rightarrow \exists \overline{e} (\set_i(\overline{e}) \land \overline{e} \in_i \overline{x} \land \add_i(\overline{e},\overline{c},\overline{d}))))
        \end{align*}
    \item $\psi_i$: the entire word is the ordered encoding of the set $V_i$. Size $\mathcal{O}(i^5)$. 
    \begin{align*}
        \psi_i := \exists x \exists y \forall z (x \leq z \land z \leq y \land V_i(x,y))
    \end{align*}
\end{itemize}

The formula $\psi_i$ defines a word $w$ that, as an encoding of the set $V_i$, has length at least $\tower(i)$. The size of $\psi_i$ is $c \cdot i^5$ for some constant $c$ so $w \in \Def(\FO[c\cdot i^5])$. As we want to relate the length of $w$ to the size of $\psi_i$, we set $n = c \cdot i^5$ and obtain the following result:

\begin{theorem}
For some constant $c \in \Nset$ there are infinitely many $n \in \Nset$ satisfying
\[
\defN(\FO[n]) \geq \tower(\sqrt[5]{n/c}).
\]
\end{theorem}

Proposition \ref{DNvsLS} immediately gives the same bound for the Hanf number.

\begin{corollary}
For some constant $c \in \Nset$ there are infinitely many $n \in \Nset$ satisfying
\[
\Hanf(\FO[n]) \geq \tower(\sqrt[5]{n/c}).
\]
\end{corollary}

By omitting the subformula $\oset_{i+1}$ from the above we get a formula of size $\mathcal{O}(i^3)$ that is no longer satisfied by only one word but still only has large models. With this formula we obtain a lower bound for the Löwenheim-Skolem number.

\begin{corollary}
For some $c \in \Nset$ there are arbitrarily large $n \in \Nset$ satisfying
\[
\LS(\FO[n]) \geq \tower(\sqrt[3]{n/c}).
\]
\end{corollary}

\section{Lower bounds for MSO}\label{mso}
In this section, we define a similar formula for MSO as we did above for FO. The formula again defines an encoding of $V_i$ but for MSO our formula is of size $\mathcal{O}(i^2)$ compared to the $\mathcal{O}(i^5)$ of FO. We achieve this by quantifying a partition of so called levels for the braces and thus the encoded sets and using a different method to define only a single encoding.

The level of the entire encoded set will be equal to the maximum depth of braces inside the set. The level of an element of a set will always be one less than the level of the parent set. This means that there will be for example empty sets with different levels in our encoding. 

We again define our formula in terms of many subformulas and briefly list the meaning and size of each subformula:
\begin{itemize}
    \item $\set_i(\overline{x})$: $\overline{x}$ encodes a set of level $i$. Size constant. 
    \begin{align*}
        \set_0(\overline{x}) &:= S(x_1, x_2) \land L(x_1) \land R(x_2) \land D_0(x_1) \land D_0(x_2) \\
        \set_i(\overline{x}) &:= x_1 < x_2 \land L(x_1) \land R(x_2) \land D_i(x_1) \land D_i(x_2) \\
        &\land \forall y (x_1 < y < x_2 \rightarrow \neg D_i(y))
    \end{align*}
    \item $\lev_i$: The relations $D_j$ define the levels of sets as intended and there are no odd braces without pairs. Size $\mathcal{O}(i^2)$.
    \begin{align*}
        \lev_i &:= \forall x (\bigvee\limits_{j = 0}^i D_j(x) \land \bigwedge\limits_{\substack{j,k \in \{0,\dots , i\} \\ j \neq k}} \neg (D_j(x) \land D_k(x))\\
        &\land \forall \overline{x} (\bigwedge\limits_{j=0}^i (\set_j(\overline{x}) \rightarrow \forall y (x_1 < y < x_2 \rightarrow \bigvee\limits_{k = 0}^{j-1} D_k(y)))) \\
        &\land \forall x_1 (\bigwedge\limits_{j=0}^i ((L(x_1) \land D_j(x_1)) \rightarrow \exists x_2 \set_j(x_1, x_2)) \\
        &\land \bigwedge\limits_{j=0}^i (R(x_1) \land D_j(x_1)) \rightarrow \exists x_2 \set_j(x_2, x_1))
    \end{align*}
    \item $\overline{x} \in \overline{y}$: $\overline{x}$ is an element of $\overline{y}$. Size constant. Assumes $\overline{x}$ and $\overline{y}$ both encode sets.
    \begin{align*}
        \overline{x} \in \overline{y} &:= y_1 < x_1 \land x_2 < y_2
    \end{align*}
    \item $\overline{x} \sim_i \overline{y}$: $\overline{x}$ and $\overline{y}$ encode the same set. Size linear in $i$. Assumes $\overline{x}$ and $\overline{y}$ encode sets of level $i$.
    \begin{align*}
        \overline{x} \sim_0 \overline{y} &:= \top \\
        \overline{x} \sim_{i+1} \overline{y} &:= \forall \overline{a}(\set_i(\overline{a}) \rightarrow \exists \overline{b} (\set_i(\overline{b}) \\
        &\land (\overline{a} \in \overline{x} \rightarrow \overline{b} \in \overline{y}) \land (\overline{a} \in \overline{y} \rightarrow \overline{b} \in \overline{x}) \land \overline{a} \sim_i \overline{b}))
    \end{align*}
    \item $\add_{i}(\overline{x}, \overline{y}, \overline{z})$: States that $\overline{x} = \overline{y} \cup \{\overline{z}\}$. Size linear in $i$. Assumes $\overline{x}$ and $\overline{y}$ encode sets of level $i$ and $\overline{z}$ encodes a set of level $i-1$.
    \begin{align*}
        \add_{i+1}(\overline{x}, \overline{y}, \overline{z}) &:= \forall \overline{a}((\set_{i}(\overline{a}) \land \overline{a} \in \overline{y}) \rightarrow \exists \overline{b}(\set_{i}(\overline{b}) \land \overline{b} \in \overline{x} \land \overline{a} \sim_{i} \overline{b})) \\
        &\land \exists \overline{c}(\set_{i}(\overline{c}) \land \overline{c} \in \overline{x} \land \overline{c} \sim_{i} \overline{z} \\
        &\land \forall \overline{d}((\set_{i}(\overline{d}) \land \overline{d} \in \overline{x} \land d_1 \neq c_1) \\
        &\rightarrow \exists \overline{e}(\set_{i}(\overline{e}) \land \overline{e} \in \overline{y} \land \overline{e} \sim_{i} \overline{d})))
    \end{align*}
    \item $V_i(\overline{x})$: $\overline{x}$ encodes the set $V_i$. Size $\mathcal{O}(i^2)$. Assumes the level partition is given.
    \begin{align*}
        V_0(\overline{x}) &:= \set_0(\overline{x}) \\
        V_{i+1}(\overline{x}) &:= \set_{i+1}(\overline{x}) \land \exists \overline{a} (\set_i(\overline{a}) \land \overline{a} \in \overline{x} \land S(a_1, a_2)) \\
        &\land \exists \overline{b} (V_i(\overline{b}) \land \overline{b} \in \overline{x} \land \forall \overline{c} \forall \overline{d}((\set_i(\overline{c}) \land \overline{c} \in \overline{x} \land \set_{i-1}(\overline{d}) \land \overline{d} \in \overline{b}) \\
        &\rightarrow \exists \overline{e} (\set_i(\overline{e}) \land \overline{e} \in \overline{x} \land \add_i(\overline{e},\overline{c},\overline{d}))))
    \end{align*}
    \item $\phi_i (x,y)$: Quantifies the level partition and states the subword from $x$ to $y$ encodes $V_i$. Size $\mathcal{O}(i^2)$.
    \begin{align*}
        \phi_i(x,y) &:= \exists D_0 \dots \exists D_i(\lev_i \land V_i(x,y)))
    \end{align*}
\end{itemize}

We now have a formula $\phi_i(x,y)$ that says the subword from $x$ to $y$ encodes the set $V_i$. There are still multiple words that satisfy this formula, since different orders of the sets and even repetition are still allowed. To pick out only one such word, we use a lexicographic order, where a shorter word always precedes a longer one.

Let $\phi'_i$ be the formula obtained from $\phi_i$ by replacing each occurrence of $L(x)$ with $P_1(x)$ and $R(x)$ with $P_2(x)$. We define the final formula $\psi_i$ of size $\mathcal{O}(i^2)$ that says the entire word model is the least word in the lexicographic order that satisfies the property of $\phi_i$. 
\begin{align*}
\psi_i &:= \exists x \exists y(\forall z (x \leq z \land z \leq y) \land \phi_i(x,y) \\
&\land \forall P_1\forall P_2 (\forall z ((P_1(z) \lor P_2(z)) \land \neg (P_1(z) \land P_2(z))) \\
&\land \forall y' ((y'<y \lor \exists z (\forall a(a < z \rightarrow (L(a) \leftrightarrow P_1(a) \land R(a) \leftrightarrow P_2(a))) \\
&\land (P_1(z) \land R(z))) \rightarrow \neg\phi'_i(x,y'))))
\end{align*}

We have used the lexicographic order here to select only one of the possible words that satisfy our property. Note that this can be done for any property. The size of such a formula will depend polynomially on the size of the alphabet, as well as linearly on the size of the formula defining the property in question.

We obtain the lower bound for the definability number as in the FO case.

\begin{theorem}
For some constant $c \in \Nset$ there are infinitely many $n \in \Nset$ satisfying
\[
\defN(\MSO[n]) \geq \tower(\sqrt{n/c}).
\]
\end{theorem}

We get the same bounds for $\LS(\MSO[n])$ and $\Hanf(\MSO[n])$ via Proposition \ref{DNvsLS}.

\begin{corollary}
For some constant $c \in \Nset$ there are infinitely many $n \in \Nset$ satisfying
\[
\LS(\MSO[n]), \Hanf(\MSO[n]) \geq \tower(\sqrt{n/c}).
\]
\end{corollary}

\section{Conclusion}

We considered the definability number, the Löwenheim-Skolem number  and the Hanf number on words in the size $n$ fragments of first-order logic and monadic second-order logic. We obtained exponential towers of various heights as upper and lower bounds for each of these numbers. 

For $\FO$, we obtained the bounds
\[
\tower(\sqrt[5]{n/c}) \leq \defN(\FO[n]) \leq \tower(n/2+\log^*((n/2)^2+n/2)+1)
\]
for some constant $c$. As corollaries, we obtained the same bounds for $\LS(\FO[n])$ and $\Hanf(\FO[n])$. In addition, by modifying the formula we used for the lower bounds, we obtained a slightly better lower bound of $\tower(\sqrt[3]{n/c})$ for $\LS(\FO[n])$.

In the case of $\MSO$, the bounds are similarly
\[
\tower(\sqrt{n/c}) \leq \defN(\MSO[n]) \leq \tower(n/2+\log^*((n/2+1)^2)+1)
\]
for a different constant $c$. We again immediately obtained the same bounds for $\LS(\MSO[n])$ and $\Hanf(\MSO[n])$.

The gaps between the lower bounds and upper bounds we have proved are quite big. In absolute terms, they are actually huge, as each upper bound is non-elementary with respect to the corresponding lower bound. However, it is more fair to do the comparison in the iterated logarithmic scale, which reduces the gap to be only polynomial. Nevertheless, a natural task for future research is to look for tighter lower and upper bounds. 

Finally, we remark that the technique for proving an exponential tower upper bound for the number of types in the quantifier rank fragments of some logic $\mathcal{L}$ is completely generic: it works in the same way irrespective of the type of quantifiers allowed in $\mathcal{L}$. Thus, it can be applied for example in the case where $\mathcal{L}$ is the extension of $\FO$ with some generalized quantifier (or a finite set of generalized quantifiers).  Assuming further that the quantifier rank fragments $L$ of $\mathcal{L}$ satisfy Theorem \ref{preservation}, we can obtain this way an exponential tower upper bound for the numbers $\defN(L)$, $\LS(L)$ and $\Hanf(L)$. On the other hand, note that if the quantifier rank fragments $L$ satisfy Theorem \ref{preservation}, then each $\equiv_L$ is an invariant equivalence relation, whence $\mathcal{L}$ can only define regular languages. Therefore it seems that our technique for proving upper bounds cannot be used for logics with expressive power beyond regular languages.

\bibliographystyle{splncs04}
\bibliography{sw}

\section{Appendix}

In this appendix, we prove the upper bounds for the numbers $N_L$ of $L$-types for both $L=\FO_k$ and $L=\MSO_k$ stated in Section \ref{number-of-types}. To do this, we need to consider the equivalence $\equiv_L$ with respect to formulas with free variables. If the number of free second-order variables is $r$, and the number of free first-order variables is $s$, this means that each word $w\in\Sigma^+$ has to be equipped with corresponding interpretations $\bar P=(P_1,\ldots,P_r)$ and $\bar p=(p_1,\ldots,p_s)$ of the variables. We call the triple $I=(w,\bar P,\bar p)$ an $(r,s)$-interpretation. For $L=\MSO_k$ and $m\in\Nset$, we define 
\[
	M_k(m)=\sum_{r+s=m}O_k(r,s),
\]
where $O_k(r,s)$, $r,s\in\Nset$ denotes the number of $\MSO_k$-types of $(r,s)$-interpretations. 

In the case $L=\FO_k$, second-order parameters $\bar P$ are not needed, and we call the pair $(w,\bar p)$ an $s$-interpretation. The number of $\FO_k$ -types of $s$-interpretations is denoted by $F_k(s)$.

\begin{lemma}\label{zero-types}
\begin{enumerate}[(a)]
\item For any $m\in \Nset$, $M_0(m)\le 2^{(m+1)^2}$.

\item For any $s\in \Nset$, $F_0(s)\le 2^{s^2+s}$.
\end{enumerate}
\end{lemma}
\begin{proof}
(a) Consider the quantifier free types of $(r,s)$-interpretations $(w,\bar P,\bar p)$. There are $2^s$ ways of choosing the letters of $\Sigma=\{l,r\}$ to the positions in $\bar p$, and at most $2^{s^2}$ ways of choosing the order\footnote{The order is a binary relation, and there are $2^{s^2}$ binary relations on a set of $s$ elements.} of the components of $\bar p$. In addition, there are $(2^r)^s = 2^{rs}$ choices regarding in which of the sets in $\bar P$ the points in $\bar p$ are in. In total there are at most $2^s \cdot  2^{s^2}\cdot 2^{rs}=2^{s(r+s+1)}$ equivalence classes of $\equiv_{\MSO_0}$ for $(r,s)$-interpretations. Thus, 
$$
	M_0(m)\le\sum_{s=0}^{m}2^{s(m+1)}
	=\frac{2^{(m+1)^2}-1}{2^{m+1}-1}
	\le 2^{(m+1)^2}, 
$$
as $2^{m+1}-1\ge 1$.

(b) Clearly $F_0(s)=O_0(0,s)\le 2^{s}\cdot 2^{s^2}=2^{s^2+s}$.
\end{proof}

\begin{lemma}\label{k-types}
\begin{enumerate}[(a)]
\item For any $k,m\in \Nset$, 
$M_{k+1}(m)\le 2^{M_k(m+1)}$.

\item For any $k,s\in \Nset$, $F_{k+1}(s)\le 2^{F_k(s+1)}$.
\end{enumerate}
\end{lemma}
\begin{proof}
(a) For $L=\MSO_k$ and $r,s\in\Nset$, let $\mathcal{C}_{r,s}$ be the set of all the $\equiv_L$ equivalence classes of $(r,s)$-interpretations. Given an $(r,s)$-interpretation $I=(w,\bar P,\bar p)$, we define 
$$
	\mathcal{A}(I)=\{C\in \mathcal{C}_{r+1,s}\mid \exists P_{r+1}\subseteq [w]: 
	(w,\bar P P_{r+1},\bar p)\in C\},
$$
and similarly
$$
	\mathcal{B}(I)=\{C\in\mathcal{C}_{r,s+1}\mid \exists p_{s+1}\in [w]: (w,\bar P,\bar p\, p_{s+1})\in C\}.
$$ 

It is now straightforward to verify that $I\equiv_{\MSO_{k+1}} I'$ for two $(r,s)$-inter\-pretations $I$ and $I'$ if and only if $\mathcal{A}(I)\cup\mathcal{B}(I)=\mathcal{A}(I')\cup\mathcal{B}(I')$.  Furthermore, if $(r,s)\not=(r',s')$, then
$\mathcal{A}(I)\cup\mathcal{B}(I)\not=\mathcal{A}(I')\cup\mathcal{B}(I')$ for any $(r,s)$-interpretation $I$ and $(r',s')$-interpretation $I'$. Indeed, if $r<r'$, then $\mathcal{B}(I)$ is a nonempty subset of $\mathcal{C}_{r,s+1}$, and $(\mathcal{A}(I')\cup\mathcal{B}(I'))\cap\mathcal{C}_{r,s+1}=\emptyset$. Similarly, if $s<s'$, then $\mathcal{A}(I)$ is a nonempty subset of $\mathcal{C}_{r+1,s}$, and $(\mathcal{A}(I')\cup\mathcal{B}(I'))\cap\mathcal{C}_{r+1,s}=\emptyset$.

Thus we see that the $\equiv_{\MSO_{k+1}}$ equivalence class of any $(r,s)$-interpretion with $r+s=m$ is uniquely determined by the set $\mathcal{A}(I)\cup\mathcal{B}(I)\subseteq\mathcal{C}^{m+1}$, where $\mathcal{C}^{m+1}$ is the union of the sets $\mathcal{C}_{u,v}$ over pairs $(u,v)$ such that $u+v=m+1$. Observe now that 
$$
	|\mathcal{C}^{m+1}|=\Bigl|\bigcup_{u+v=m+1}\mathcal{C}_{u,v}\Bigr|=
	\sum_{u+v=m+1}|\mathcal{C}_{u,v}|=M_k(m+1),
$$
as clearly $|\mathcal{C}_{u,v}|=O_k(u,v)$. Thus we obtain the desired upper bound:
$$
	M_{k+1}(m)\le |\mathcal{P}(\mathcal{C}^{m+1})|=2^{|\mathcal{C}^{m+1}|}=2^{M_k(m+1)}.
$$

(b) The proof is similar to that of (a).
\end{proof}

\begin{customthm}{2}
For any $k\in\Nset$, 
$N_{\MSO_k}\le \tower(k+\log^*((k+1)^2)+1)$\\ and
$N_{\FO_k}\le \tower(k+\log^*(k^2+k)+1)$.
\end{customthm}
\begin{proof}
We prove by induction on $k$ that, for all $m\in\Nset$, 
$$
	M_k(m)\le \tower(k+\log^*((k+m+1)^2)+1).
$$
In the case $k=0$ this follows from Lemma \ref{zero-types}, since 
$$
	2^{(m+1)^2}\le \tower(\log^*((m+1)^2)+1).
$$

Assume then as an inductive hypothesis that 
$$
	M_k(m+1)\le \tower(k+\log^*((k'+1)^2)+1)
$$
for $k'=k+m+1$. Then by Lemma \ref{k-types} we get 
\begin{align*}
	M_{k+1}(m)&\le 2^{\tower(k+\log^*((k'+1)^2)+1)}\\
	&=\tower(k+1+\log^*((k'+1)^2)+1)\\
	&=\tower(k+1+\log^*((k+1+m+1)^2)+1),
\end{align*}
as desired. 

Note that $N_{\MSO_k}=M_k(0)$. Thus, applying the inequality above for $m=0$, we obtain $N_{\MSO_k}\le \tower(k+\log^*((k+1)^2)+1)$. 

The second claim is proved in the same way.
\end{proof}

\end{document}